\documentclass[11pt,a4paper,english]{article}%
\usepackage[T1]{fontenc}
\usepackage[latin9]{inputenc}
\usepackage{amsthm,amsmath,amssymb,amsfonts}
\usepackage{babel}
\usepackage{graphicx}
\usepackage{lmodern}
\usepackage{natbib}

\usepackage{geometry}
\usepackage[usenames]{color}

\usepackage{tabularx}
\newcolumntype{R}{>{\raggedright\arraybackslash}X}

\def\ud{\mathrm{d}}

\newtheorem{theorem}{Theorem}
\newtheorem{proposition}{Proposition}

\usepackage{hyperref}%
\hypersetup{
pdftitle                = {Explicit investment rules with time-to-build and uncertainty},
pdfauthor               = {Ren\'e Aid; Salvatore Federico; Huy\^en Pham; Bertrand Villeneuve},
pdfdisplaydoctitle      = true,
colorlinks              = true,
linkcolor               = blue,
urlcolor                = blue,
citecolor               = blue,
pdfdisplaydoctitle      = true,
pdftoolbar              = true,
pdfmenubar              = true,
pdffitwindow            = true}

\title{Explicit investment rules with time-to-build and uncertainty%
\footnote{This study was supported by FiME (Laboratoire de Finance des March\'es de l'Energie) and
the ``Finance et D\'eveloppement Durable - Approches Quantitatives'' Chair. }}
\author{
Ren\'e Aid%
\footnote{EDF R\&D and www.fime-lab.org. rene.aid@edf.fr}
\qquad
Salvatore Federico%
\footnote{Universit\`a degli Studi di Milano. salvatore.federico@unimi.it}
\qquad
Huy\^en Pham%
\footnote{LPMA, Universit\'e Paris-Diderot. pham@math.univ-paris-diderot.fr}
\qquad
Bertrand Villeneuve\footnote{LEDa, Universit\'e Paris-Dauphine. bertrand.villeneuve@dauphine.fr}}

\begin{document}

\maketitle

\begin{abstract}
We establish explicit socially optimal rules for an irreversible investment decision
with time-to-build and uncertainty.
Assuming a price sensitive demand function with a random intercept,
we provide comparative statics and economic interpretations
for three models of demand
(arithmetic Brownian, geometric Brownian, and the Cox-Ingersoll-Ross).
Committed capacity, that is, the installed capacity plus the investment in the pipeline,
must never drop below the best predictor of future demand, minus two biases.
The discounting bias takes into account the fact that investment is paid upfront for future use;
the precautionary bias multiplies a type of risk aversion index by the local volatility.
Relying on the analytical forms, we discuss in detail the economic effects.
\end{abstract}

\smallskip


\smallskip


\smallskip



\section{Introduction}

How to track demand when the time-to-build retards capacity expansion?
When to invest and how much?
We answer these questions with a model of irreversible investment. In this model, the objective
of the decision-maker is to minimize the total discounted social cost. 
That is, the decision-maker maximizes the standard microeconomic social surplus.
We are able to show in particular that the solution is implementable as a competitive equilibrium.
We are able to calculate explicit, compact, decision rules.

In many capitalistic industries, construction delays are essential. 
In this paper, we focus on electricity generation. In this sector, construction delays can be considerable: 
they could be only one year for a small wind-farm 
but could be three years for a gas turbine 
and eight to ten years for a nuclear plant.
The scenarios of the evolution of demand 
with their trends, their drag force, and their stochastic parts require
particular attention.
To this purpose, 
we develop the comparative statics and economic interpretations for three demand models applied to electricity generation.
The intercept of the price sensitive demand function follows
either an arithmetic Brownian motion as in \citet{Bar-Ilan02},
or a geometric Brownian motion as in \citet{Bar-Ilan96} and \citet{Aguerrevere03},
or the Cox-Ingersoll-Ross (CIR) model.
The latter is a mean-reverting process, and, to our knowledge, 
no real options investment model exists 
in the literature with time-to-build and a process of this type.
The basic existence and regularity results are provided in
a companion paper (\citet{Federico14}), 
but we simplify the specification for the sake of calculability.

An exact decision rule facilitates the clear understanding of the effects at play.
The decision rule stipulates what the committed capacity should be, 
that is, the installed capacity plus capacity under construction. 
The action rule, given the current conditions, is that the committed capacity must
not fall below
the best predictor of demand after the delay, minus two biases.
The first bias is a pure discounting bias unrelated to uncertainty: 
because the investment is paid for upfront but only produces after the delay,
the required committed capacity is reduced.
The second one is a precautionary bias where a risk aversion index is
multiplied by local volatility .

Because of this decomposition, we find the following qualitative results.
For investment with a long construction horizon, we find that
uncertainty has very little importance compared to the trend of demand growth.
Thus, an error on the long-term trend is much more harmful than an error on volatility.
For investment with a short construction horizon, the opposite is true:
decision-makers should pay greater attention to volatility.
We also illustrate the practical importance of a possible saturation of the demand
with the CIR model.
The investors' behavior is very different depending on whether demand
is above or below the long-run average, or target.
When demand is above the target, the investor is almost insensitive to the current demand, 
except if the return speed is very slow.
Below the target, the comparison between the time-to-build and the expected time-to-target
is critical: if the time-to-build is longer, then the optimal committed capacity is practically the
target itself minus the biases; if the time-to-build is shorter, then the investors observe the process and invest progressively.

The literature on the topic provides a number of insights.
Table~\ref{tab:literature} provides a tentative classification.
The competitive pressure matters: competition kills the value of waiting and thus
accelerates investment. 
\citet{Grenadier00,Grenadier02} and \citet{Pacheco03} follow this line of thought.
We exclusively use a competitive market and show that this effect is completely internalized.
The seminal work \citet{McDonald86} on the option to wait in the case of irreversible decisions
shows that uncertainty has a negative effect on investment.
Strong support for this result is that
with greater volatility, investment is triggered by a higher current product price, i.e. a smaller probability of a market downturn.
Several extensions provide conditions under which this result does not hold 
or might be mitigated.
Construction delays, 
that is, the time between the decision and the availability of the new capacity, 
have attracted the attention of economists.
In particular, the models
in \citet{Bar-Ilan96}, \citet{Bar-Ilan02}, and \citet{Aguerrevere03}
exhibit situations where an increase in uncertainty leads to an increase in investment.

The models that exhibit a positive effect on investment from an increase in
uncertainty, do so only for a specific range of parameters.
Besides, the quantitative effects are very small.
\citet{Bar-Ilan02} show
in their simulations that when the uncertainty on demand is multiplied by five, then
the investment threshold moves only by 1\%. And as the authors themselves point out,
the investment thresholds are nearly independent of the level of uncertainty.
The large effects found in \citet{Majd87} are reconsidered in \citet{Milne2000}.

In \citet{Aguerrevere03}, a paper with which we share most of the modeling choices,
the production is flexible, although the capacity accumulation is not.
Investors keep the choice to produce only when it is profitable, and thus
the rigidity of investment is attenuated by the option to produce or not.
The capacity reserves are all the more profitable the longer the time-to-build.
In consequence, uncertainty tends to increase the investment rate.
This paper is significant because of the way it integrates meaningful economic questions,
and the numerical simulations are instructive.

As far as electricity production is concerned, the flexibility of the base production
is limited either for technological reasons (nuclear plants) or because
the fixed cost per idle period are important (coal- or gas-fired power plants).
In which case, the cost difference between producing or not is narrow.
Our approach fills a gap in the literature.

\begin{table}
\begin{center}
{\small
\begin{tabular}{lccc}
Paper                  &Objective      &Competition         & Investment          \\ \hline
                       &               &                    &                     \\
\citealt{Majd87}       &firm           &no                  & irreversible        \\
\citealt{Bar-Ilan96}   &firm           &no                  & reversible          \\
\citealt{Grenadier00}  &firm           &perfect             & irreversible        \\
\citealt{Bar-Ilan02}   &planner        &no                  & irreversible        \\
\citealt{Grenadier02}  &firm           &imperfect           & irreversible        \\
\citealt{Aguerrevere03}&planner/firm   &perfect/imperfect   & irreversible with   \\
                       &               &                    &flexible production     
\end{tabular}
} 
\end{center}
\caption{Papers on investment with uncertainty and time-to-build.}
\label{tab:literature}
\end{table}

This paper is organized as follows.
Section \ref{sec:model} describes and justifies our modeling approach.
Solutions and general properties are provided in Section~\ref{sec:solution}.
We give the expression of the decision rule and we show that
the solution to the optimization program can be decentralized as a competitive equilibrium.
The economic analysis of the joint effect of time-to-build  and uncertainty is given
in Section \ref{sec:GBM} for the a geometric Brownian motion, and
in Section~\ref{sec:CIR} for the CIR model.
Section \ref{sec:conclusion} concludes.

For information on the popular arithmetic Brownian motion application, please see Appendix \ref{sec:ABM}.

\section{The model}              \label{sec:model}

We set up a model of an irreversible investment decision in which the objective
is to track the current demand of electricity as closely as possible.

\paragraph{1.} The inverse demand function at date $t$ is
\begin{align}             \label{eq:demand-price}
p_t(Q) = \eta + \theta (D_t - Q),
\end{align}
with $\eta\geq0$ and $\theta>0$, where $p$ is the price and $Q$ is the output.%
\footnote{\citet{Aguerrevere03} takes a similar form and discusses its flexibility.}
The  (quasi) intercept $(D_t)_{t\geq 0}$ is a diffusion that satisfies the SDE
\begin{align}     \label{eq:demand}
\begin{cases}
\ud D_t=\mu (D_t){{\ud}t}+\sigma (D_t)\ud W_t,\\
D_0=d,
\end{cases}
\end{align}
where $(W_t)_{t\geq 0}$ is a Brownian motion on some filtered probability space
$(\Omega,\mathcal{F},(\mathcal{F}_t)_{t\geq 0}, \mathbb{P})$.
Without loss of generality, we suppose that  the filtration $(\mathcal{F}_t)_{t\geq 0}$
is the one generated by the Brownian motion $W$ and enlarged by the $\mathbb{P}$-null sets.

\paragraph{2.}
There is a time lag $h>0$ between the date of the investment
decision and the date when the investment is completed and becomes productive.
Thus, the investment  decision at time $t$ brings additional capacity at time $t+h$.

\paragraph{3.}
At time $t=0$, there is an initial stream of pending investments initiated in the interval $[-h,0)$
that are going to be completed in the interval $[0,h)$.
The function that represents the cumulative investment planned
in the interval $[-h,s]$, $s\in (-h,0)$, is  a nonnegative non-decreasing c\`adl\`ag function.
Therefore, the set where this function lives is
\begin{align}
\mathcal{I}^0=\{I^0:[-h,0)\rightarrow \mathbb{R}^+, \ s\mapsto I^0_s \mbox{ c\`adl\`ag, non-decreasing}\}.
\end{align}
We set
\begin{align}
 I^0_{0^-}=\lim_{s\uparrow 0}I^0_s,\qquad I^0\in\mathcal{I}^0.
\end{align}

\paragraph{4.}
The decision variable is represented by a c\`adl\`ag non-decreasing
$(\mathcal{F}_t)_{t\geq 0}$-adapted process $(I_t)_{t\geq 0}$
that represents the cumulative investment from time $0$ up to time $t\geq 0$.
Therefore, formally $\ud I_t$ is the investment at time $t\geq 0$.
The set of admissible strategies, which we denote by $\mathcal{I}$, is  the set
\begin{align}
\mathcal{I}= \{I:\mathbb{R}^+\times \Omega\rightarrow \mathbb{R}^+,
             \ I \ \mbox{c\`adl\`ag, $(\mathcal{F}_t)_{t\geq 0}$-adapted, nondecreasing}\}.
\end{align}

\paragraph{5.}
Due to the considerations above, given $I^0\in\mathcal{I}^0$, $I\in \mathcal{I}$ and setting $I_{-h^-}=0$,
we assume that the production capacity $(K_t)_{t\geq 0}$
follows the controlled dynamics driven by the state equation
\begin{align}\label{statecontrol}
\begin{cases}
\displaystyle{\ud K_t=\ud I_{t-h}}, & \\
K_{0^-}=k,     \qquad I_s=I^0_s,\ s\in[-h,0).
\end{cases}
\end{align}
The equation above is a controlled locally deterministic differential equation
with delay in the control variable.
By  solution to this equation, we mean the the c\`adl\`ag process:
\begin{align}
 K_t=k+\bar{I}_{t-h}, \quad\forall t\geq 0,
\end{align}
 where $\bar{I}$ is the process
\begin{align}
\bar{I}_t=\begin{cases}
        I^0_t,               &  t\in[-h,0),\\
        I^0_{0^-}+I_t, \quad &  t\geq 0.
\end{cases}
\end{align}
Significantly, the randomness in \eqref{statecontrol} enters only through $I$,
and there are no stochastic integrals.

\paragraph{6.}
The objective is to minimize over  $I\in \mathcal{I}$ the functional
\begin{align}            \label{functional}
F(k,d,I^0;I)=\mathbb{E}
\left[\int_0^{+\infty} e^{-\rho t} \left(\frac{1}{2}(K_t-D_t)^2{{\ud}t}+q_0\ud I_t\right)\right],
\end{align}
where $q_0>0$ is the unit investment cost.

\bigskip

\paragraph{Economic interpretation of the loss function.}
The program is a maximization of the social surplus or conversely the minimization of the deadweight loss.
Given \eqref{eq:demand-price}, 
the instantaneous net consumers' surplus is by standard definition:
\begin{align}
S_t =  \int_{0}^{K_t} \left( \eta+ \theta(D_t - q) \right) \, dq - p_t K_t.
\end{align}
This is the sum of the values given to each unit consumed minus the price paid for them.
If the unit production cost is $\eta$ and if there is some fixed cost $F$ per year,
the instantaneous producer's profit $\pi_t$ is $(p_t - \eta) K_t - F$.
The total instantaneous surplus $\text{TS}_t =   S_t + \pi_t $ is therefore:
\begin{align}
\text{TS}_t = {} & \theta  \int_{0}^{K_t} (D_t - q) \, dq  - F \\
   = {} &
\underbrace{- \frac{\theta }{2 } \left(K_t - D_t  \right)^2}_{\mbox{Depends on control}}
+ \underbrace{ \frac{\theta}{2 } D_t^2 - F}_{\mbox{Doesn't}} .
\end{align}
Maximizing the discounted social surplus minus the investment costs
amounts to program \eqref{functional}, where the true investment cost $q_0$ is divided by
$\theta$.

Thus, the solution and \eqref{eq:demand-price} generate an electricity price process.
It reflects the marginal cost plus a term, which can be negative,
that reflects tension in the market.

\paragraph{Diffusion process.}
The process $D$ satisfies the following conditions:%
\footnote{A reference for the theory of one-dimensional diffusions is \citet{Karatzas91}.}
we assume that the coefficients $\mu,\sigma:\mathbb{R}\rightarrow\mathbb{R}$ in \eqref{eq:demand} are continuous
with sublinear growth and regular enough to ensure the existence of a unique strong solution to \eqref{eq:demand}.
Further, we assume that this solution takes values in an open set $\mathcal{O}$ of $\mathbb{R}$
and that it is non-degenerate over this set, that is, $\sigma^2>0$ on $\mathcal{O}$.
In the example we shall discuss in the next section, the set $\mathcal{O}$ will be $\mathbb{R}$ or $(0,+\infty)$.
We observe that, due to the assumption of sublinear growth of $\mu,\sigma$,
standard estimates in SDEs (see, e.g., \citet[Ch.\,II]{Krylov80})
show that there exist  $\kappa_{0}, \kappa_1$ depending on $\mu,\sigma$  such that
\begin{align}            \label{growth2}
\mathbb{E}\left[|D_t|^2\right]\leq \kappa_{0}(1+|d|^2)e^{\kappa_{1} t},
\qquad t\geq 0.
\end{align}

\section{Solution}             \label{sec:solution}

The problems with delay are by nature of infinite dimension.
Referring to our case, the functional $F$ defined in \eqref{functional} depends not only on the initial $k$
but also on the past  of the control $I^0$, which is a function.
Nevertheless, the problem can be reformulated
in terms of another one-dimensional state variable not affected by the delay.
We rewrite the objective functional to introduce a new state variable,
the so-called \emph{committed capacity}.

The idea of the reformulation in control problems with delay is contained in \citet{Bar-Ilan02}
(cf.\ also \citet{Bruder2009}) 
in the context of optimal stochastic impulse problems.
Here, we develop this idea for singular stochastic control.
It is worth stressing that, unlike \citet{Bar-Ilan02},
we simplify the approach by working not on the value function of the optimization problem
but directly on the basic functional.

\subsection{Reduction to a problem without delay}
For the case of the domain for the couple of variables $(k,d)$ of our problem, the set is:
\footnote{The real problem is meaningful for $k\geq 0$;
nevertheless, it is convenient from the  mathematical point of view to allow the case of $k<0$.
Because the problem is irreversible and starts from $k\geq 0$, the capital remains nonnegative.}
\begin{align}
{\mathcal{S}}=\mathbb{R}\times {\mathcal{O}}.
\end{align}
Define the  committed capacity as:
\begin{align}
C_t:=K_t+\bar{I}_t-\bar{I}_{t-h}=K_{t+h}.
\end{align}
In differential form, the dynamics of $C_t$ is
\begin{align}            \label{committed}
\begin{cases}
dC_t=\ud I_t,\\
C_{0^-}=c=k+I^0_{0^-}.
\end{cases}
\end{align}
Therefore, it does not contain the delay in the control variable.

From now on, the dependence of $K$ on $k, I^0, I$; the dependence of $C$ on $c,I$;
and the dependence of $D$ on $d$ is denoted respectively as $K^{k,I^0,I}$, $C^{c,I}$, and $D^d$.
\smallskip

The crucial facts that allow the removal of the delay are the following.
\begin{enumerate}
\item The committed capacity is $(\mathcal{F}_t)_{t\geq 0}$-adapted.
This is due to the special structure of the controlled dynamics of $K$
that makes $K_{t+h}^{k,I^0,I}$ known given the information  $\mathcal{F}_t$.
\item Within the interval $[0,h)$, the control $I$ does not affect the dynamics of $K^{k,I^0,I}$,
which is (deterministic and) fully determined by $I^0$.
In other words, $K^{k,I^0,I^{(1)}}_t= K^{k,I^0,I^{(2)}}_t$ for every $t\in[0,h)$ and every $I^{(1)},I^{(2)}\in \mathcal{I}$.
Therefore, we can write without ambiguity $K^{k,I^0}_t$ for $t\in[0,h)$
to refer to the ``controlled'' process $K$ within the interval $[0,h)$.
\end{enumerate}
Given these observations, we have the following:
\begin{proposition}                   \label{prop:reduction}
\begin{align}
F(k,d,I^0;I) =
 \mathbb{E}\left[\int_0^{+\infty} e^{-\rho t}
    \big(g(C_t^{c,I},D_t^d){{\ud}t}+q_0\ud I_t\big)\right] +J(k,d,I^0),      \label{fun2}
\end{align}
where
\begin{align}
J(k,d,I^0) = \frac{1}{2}\,\mathbb{E}\left[\int_0^{h} e^{-\rho t}
       \left(K_t^{k,I^0}-D_t^d\right)^2{{\ud}t}\right],
\end{align}
and $g:\mathcal{S}\rightarrow\mathbb{R}^+$ is defined by
\begin{align}
g(c,d):&=\frac{1}{2}e^{-\rho h}\mathbb{E}\left[(c-D^d_h)^2\right] 
         \nonumber \\
       &=\frac{1}{2}e^{-\rho h}(c^2- 2\beta_0(d)c+\alpha_0(d)), 
         \label{ggg}
\end{align}
where
\begin{align}               \label{beta0}
  \alpha_0(d):=&\mathbb{E}\big[\big|D_h^d\big|^2\big], \qquad \beta_0(d):=\mathbb{E} \big[D_h^d\big].
\end{align}
\end{proposition}

\begin{proof}
Using the definition of $g$, the time-homogenous property of $D$, we have:
\begin{align}
\mathbb{E}\left[g(C_t^{c,I},D_t^d)\right]
&= \frac{1}{2}e^{-\rho h}\mathbb{E}\left[\mathbb{E}\left[(c'-D_h^{d'})^2\right]\Big|_{c'=C_t^{c,I}, \ d'=D_t^d}\right] \nonumber\\
&= \frac{1}{2}e^{-\rho h}\mathbb{E}\left[\mathbb{E}\left[(C_t^{c,I}-D^d_{t+h})^2\, | \, \mathcal{F}_t\right]\right] \nonumber \\
&= \frac{1}{2}e^{-\rho h}\mathbb{E}\left[(C_t^{c,I}-D^d_{t+h})^2\right]  \nonumber \\
&= \frac{1}{2}e^{-\rho h}\mathbb{E}\left[(K_{t+h}^{k,I^0,I}-D^d_{t+h})^2\right]. 
\end{align}
Therefore, \eqref{functional} can be rewritten as
\begin{align}
F(k,d,I^0;I) =
{} & \mathbb{E}\left[\int_{[0,h)} e^{-\rho t}
     \left(\frac{1}{2}\left(K_t^{k,I^0,I}-D_t^d\right)^2{{\ud}t}+q_0\ud I_t\right)\right]   \nonumber \\
     {} & +  \mathbb{E}\left[\int_{[h,+\infty)} e^{-\rho t}
     \left(\frac{1}{2}\left(K_t^{k,I^0,I}-D_t^d\right)^2{{\ud}t}+q_0\ud I_t\right)\right]   \nonumber \\
     = {}&   \mathbb{E}\left[\int_{[0,h)} e^{-\rho t}
     \left(\frac{1}{2}\left(K_t^{k,I^0,I}-D_t^d\right)^2{{\ud}t}+q_0\ud I_t\right)\right]   \nonumber\\
{} & +\mathbb{E}\left[\int_0^{+\infty} e^{-\rho (t+h)}
   \left(\frac{1}{2}\left(K_{t+h}^{k,I^0,I}-D_{t+h}^d\right)^2{{\ud}t}+q_0\ud I_{t+h}\right)\right] \nonumber \\
= {} &
       \mathbb{E}\left[\int_0^{+\infty} e^{-\rho t}
       \big(g(C_t^{c,I},D_t^d){{\ud}t}+q_0\ud I_t\big)\right]+ J(k,d,I^0).    
\end{align}
\end{proof}
Thus, the functional $J(k,d,I^0)$ defined in Proposition \ref{prop:reduction} does not depend on $I\in\mathcal{I}$.
Therefore, by setting
\begin{align}
  G(c,d;I):=\mathbb{E}\left[\int_0^{+\infty}e^{-\rho t} \big(g(C_t^{c,I},D_t^d)+q_0\ud I_t\big)\right],
\end{align}
the original optimization problem of minimizing $F(k,d,I^0;\cdot)$ over $\mathcal{I}$
is equivalent to the optimization problem  \emph{without delay}
\begin{align}    \label{committedfun}
v(c,d):=   \inf_{I\in\mathcal{I}}\ G(c,d;I) \ \mbox{ subject to }  (\ref{committed})  \mbox{ and }  (\ref{eq:demand}).
\end{align}

\subsection{Solution characterization}

In the sequel, to give sense to the problem (i.e., to guarantee finiteness),
we make the standing assumption that the discount factor $\rho$ satisfies
\begin{align}    \label{rhoK1}
\rho\;  > \; \max (\kappa_1,0),
\end{align}
where  $\kappa_{1}$ is the constant appearing in \eqref{growth2}.
This assumption guarantees that there is some $\kappa$ depending on $\mu,\sigma$ s.t.
\begin{align}
0 \; \leq \; v(c,d)  \; \leq \; \kappa \ (1+|c|^2+|d|^2),
\qquad \forall (c,d)\in \mathcal{S}.
\end{align}
In particular, it implies that the value function $v$ is finite and locally bounded.

\citet{Federico14} prove the following facts.%
\footnote{\citet{Federico14} deal with reversible problems.
We can apply their results by taking an infinite cost of disinvestment.
The irreversible case with a profit maximization criterion is studied with similar generality 
in \citet{Ferrari}.}

\begin{enumerate}
\item $v$ is convex with respect to the variable $c$
\item $v$ is differentiable with respect to $c$, and  $v_c$ is continuous in $\mathcal{S}$
\item The function $d\mapsto v_c(c,d)$ does not increase for each $c\in\mathbb{R}$
\item $v_c\geq - q_0$
\end{enumerate}
 In view of these facts, there is now the  \emph{continuation region}
\begin{align}
\mathcal{C}:=\{(c,d)\in {\mathcal{S}} \ | \ v_c(c,d)>-q_0\},
\end{align}
and  the   \emph{action region}
\begin{align}
\mathcal{A}:=\{(c,d)\in{\mathcal{S}} \ | \ v_c(c,d)=-q_0\}.
\end{align}
Therefore, $\mathcal{C}$ and $\mathcal{A}$ are disjoint and $\mathcal{S}=\mathcal{C}\cup\mathcal{A}$.
Due to the continuity of $v_c$, the continuation region is an open set of $\mathcal{S}$,
while the action region is a closed set of $\mathcal{S}$.
Moreover, due to the monotonicity of $v_c(c,\cdot)$ and to the convexity of $v(\cdot,d)$,
$\mathcal{C}$ and $\mathcal{A}$  can be rewritten as
\begin{align}
 \mathcal{C}=\{(c,d) \in {\mathcal{S}} \ | \ c>\hat{c}(d)\}, \quad
 \mathcal{A}=\{(c,d) \in {\mathcal{S}} \ | \ c\leq \hat{c}(d)\},
\end{align}
where $\hat{c}:\mathcal{O}\rightarrow \mathbb{R}$ is a non-decreasing function.
The latter function is the \emph{optimal boundary} for the problem in the sense
that it characterizes the optimal control.
Thus, in this singular stochastic optimal control,
the optimal control consists of keeping the state processes
within the closure of the continuation region $\mathcal{C}$ by reflecting the controlled process
on the optimal boundary along the direction of the control.
See Figure \ref{fig:CA-in-dc}.

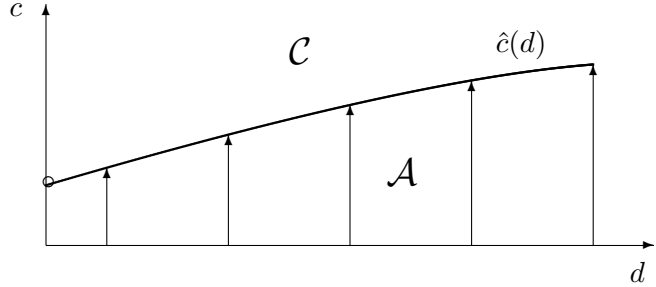
\begin{figure}[ht!]
\begin{center}
\setlength{\unitlength}{1.6cm}
\begin{picture}(6,2.5)(-3,-0.8)
\put(-2.5,-.5){\vector(1,0){5}}
\put(2.3,-0.8){$d$}
\put(-.5,1){$\mbox{\Large{$\mathcal{C}$}}$}
\put(.3,0){$\mbox{\Large{$\mathcal{A}$}}$}
\put(-2.54,-.035){$\circ$}
\put(-2.5,-.5){\vector(0,1){2}}
\put(-2,-0.5){\vector(0,1){0.65}}
\put(-1,-0.5){\vector(0,1){0.92}}
\put(-0,-0.5){\vector(0,1){1.17}}
\put(1,-0.5){\vector(0,1){1.37}}
\put(2,-0.5){\vector(0,1){1.5}}
\put(1.2,1.1){$\hat{c}(d)$}
\put(-2.8,1.4){$c$}
\thicklines
\qbezier(-2.5,-.001)(0.5,0.8853)(2,1)
\end{picture}
\end{center}
\caption{Continuation region ($\mathcal{C}$) and action region ($\mathcal{A}$)
         in the demand-committed capacity space.}
\label{fig:CA-in-dc}
\end{figure}

We have an explicit characterization of $\hat{c}$, that is, of the optimal control
that is provided by the following result.

\begin{theorem}       \label{prop:structure}
The optimal boundary is explicitly written as
 \begin{align}        \label{optbound}
\hat{c}(d)=\beta_0(d)
           - q_0 \rho e^{\rho h}
           +\frac{1}{2}\sigma^2(d)\,\frac{\beta''(d)\psi'(d)-\beta'(d)\psi''(d)}{\psi'(d)},
 \end{align}
where $\beta_0(d)$ is defined in \eqref{beta0} as $\mathbb{E} \big[D_h^d\big]$,
 \begin{align}
  \beta(d):=\int_0^{+\infty}e^{-\rho t}\mathbb{E}[\beta_0(D_t^d)]{{\ud}t},
 \end{align}
and $\psi$ is the strictly  increasing fundamental solution
to the linear ODE
\begin{align}          \label{ODE}
[\mathcal{L}\phi](d):=
  \rho \phi(d)-\mu(d)\phi'(d)-\frac{1}{2}\sigma^2(d) \phi''(d)=0, \qquad d\in\mathcal{O}.
\end{align}

The unique optimal control for the problem \eqref{committedfun} is the process
\begin{align}
I^*_t=\left[\hat{c}\left(\sup_{0\leq s\leq t} D_s^d\right)-c\right]^+.
\label{eq:boundary}
\end{align}
\end{theorem}

\begin{proof}
Theorem 4.2 and Corollary 5.2 of \citet{Federico14} state the above claim%
\footnote{Note however that here we have the term $e^{\rho h}$ multiplying $q_0$.
This is due to the fact that our function $g$ is equal to the function $g$
in Section 5 of \citet{Federico14} up to the constant $e^{-\rho h}$.}
with
\begin{align}        \label{hatcmr}
\hat{c}(d)=\rho \left[\beta(d)-\frac{\psi(d)}{\psi'(d)}\beta'(d)-q_0e^{\rho h}\right],
\end{align}
Therefore, if \eqref{hatcmr} can be rewritten in the form \eqref{optbound},
then it is more suitable for interpretation.

To this purpose, because $\psi$ solves the  ODE \eqref{ODE}, we have
\begin{align}         \label{e1}
\hat{c}(d) = \rho \beta(d)-\mu(d)\beta'(d)-\frac{1}{2}\sigma^2(d)\frac{\psi''(d)}{\psi'(d)}\beta'(d)- q_0 \rho e^{\rho h}.
\end{align}
On the other hand, it is well-known from the connection between the linear ODE
and the one-dimensional diffusions that the function $\beta$
solves the nonhomogeneous ODE \eqref{ODE} with the forcing term $\beta_0$:
\begin{align}          \label{e2}
\mathcal{L}\beta=\beta_0.
\end{align}
Hence, combining \eqref{e1} and \eqref{e2}, the expression \eqref{optbound} follows.
\end{proof}

The socially optimal investment as calculated above is also,
given the price process,
a profit-maximizing investment for price-taking investors.
Therefore, the optimum can be decentralized as a competitive equilibrium:
\begin{proposition} \label{prop:competitive}
Let $p^{c,d,*}$ be the price process at the optimum.
We have
\begin{align}
\mathbb{E}
\left[
\int_0^{+\infty} e^{-\rho t}\left( p^{c,d,*}_t  -\eta \right) \ud t
\right] \leq  q_0,
\end{align}
with the equality holding if and only if $(c,d)\in\mathcal{A}$.
\end{proposition}

More precisely, investment is null if the expected present revenue
from the additional unit is strictly lower than its cost,
whereas all profitable opportunities are exhausted for the case of equality.
The proof is in Appendix \ref{sec:competitive}.

\subsection{Interpretation of the boundary}        \label{sub:explicit}

The boundary $\hat{c}(d)$ defined by \eqref{optbound} and the optimal control
defined by \eqref{eq:boundary} are easily amenable to interpretations.
The boundary is composed of three terms:
\begin{align}
\hat{c}(d) &= \beta_0(d) - b_{\rho} - b_{\sigma}(d).
\end{align}
\begin{enumerate}
\item $\beta_0(d)$ is what $d$ is expected to be $h$ years later:
one commits to what demand is expected to be when the investment becomes operative.

\item The discounting bias $b_{\rho}= q_0 \rho e^{\rho h}$ 
expresses the fact that the investment is paid right away, 
whereas the cost of the insufficient capacity is discounted.

This effect can be retrieved with a heuristic non-stochastic version of the model.
Denote by $\Delta$, the permanent downward shift in capacity, compared to the best estimate
$\beta_{0}(d)$.
The investor permanently suffers the loss 
$\frac{1}{2}  \Delta^2$ 
per year in which the total actuarial cost is
$\frac{1}{2} \frac{\Delta^2}{\rho}$.
The total money saved by shifting capacity is $q_0 \Delta$.
The investor minimizes
\begin{align}
\frac{1}{2}  e^{-\rho h} \frac{\Delta^2}{\rho} - q_0 \Delta
\end{align}
with respect to $\Delta$.
The minimizing $\Delta$ is $q_0 \rho e^{\rho h}$.

\item The precautionary bias 
\begin{align}  \label{chatsigma}
b_{\sigma}(d)  
:=\frac{1}{2}\sigma^2(d)\, \left[\beta'(d) \frac{\psi''(d)}{\psi'(d)} -\beta''(d)\right].
\end{align}
gives the security margin due to the stochastic nature of the demand process. 
It is null if, for example, $\sigma(d)=0$.

The calculations go one step further if we assume the affine drift $\mu(d)=ad+b$.
Then we have
\begin{align}
\beta_0(d)= de^{ah} - bh \frac{1-e^{ah}}{ah}.
\end{align}
The ratio must be taken as $-1$ when $a=0$.
Therefore, $\beta''=0$ in this case, and
\begin{align}
b_{\sigma}(d) =
\frac{1}{2} \sigma^2(d) \frac{e^{ah}}{\rho -a} \frac{\psi''(d)}{\psi'(d)}.
\end{align}
\end{enumerate}

For the latter term $b_{\sigma}(d)$:
\begin{itemize}
\item The delay has an impact only if $a\neq 0$.
The sign of $a$ determines the impact of the delay:
the uncertainty about the future grows (diminishes) when $h$ increases if $a>0$ ($a<0$),
which justifies a bigger (smaller) bias.

\item The factor $\sigma^2(d)$ is local, it takes into account the local risk only.

\item The factor $\frac{\psi''(d)}{\psi'(d)}>0$ takes into account the global risk.%
\footnote{\citet[Prop. (50.3), Ch. V (p.292)]{Rogers2000} show that $\psi$ strictly increases and is convex.}
This term is a kind of absolute risk aversion related to the dynamics of $D$, not the delay.
\end{itemize}

\section{Geometric Brownian Motion}  \label{sec:GBM}

\subsection{The boundary}

In the case where the demand follows a geometric Brownian motion (GBM):
\begin{align}
\ud D_t=\mu D_t{{\ud}t}+\sigma D_t\ud W_t, \quad \mu\in\mathbb{R},\ \sigma>0,
\end{align}
with initial datum $d>0$,
the minimal constant $\kappa_1$ for which \eqref{growth2} is verified is $2\mu+\sigma^2$.
Therefore, according to \eqref{rhoK1}, we assume that
\begin{align}        \label{ass:rho}
 \rho>2\mu+\sigma^2.
\end{align}
In this case $\mathcal{O}=(0,+\infty)$ and
\begin{align}
\beta_0(d)= e^{\mu h}d
\quad \text{ and } \quad
\beta(d)= \frac{e^{\mu h}}{\rho-\mu} d.
\end{align}
Moreover,
\begin{align}
 [\mathcal{L}\phi](d)=\rho \phi(d)-\mu d\phi'(d)-\frac{1}{2}\sigma^2 d^2\phi''(d),
     \quad \phi\in C^2(\mathcal{O};\mathbb{R}),
\end{align}
and the fundamental increasing solution to $\mathcal{L}\phi=0$ is
\begin{align}
\psi(d)= d^{m},
\end{align}
where $m$ is the positive root of the equation
\begin{align}
 \rho-\mu m-\frac{1}{2}\sigma^2 m(m-1)=0.
\end{align}
Due to Theorem \ref{prop:structure}, we have
\begin{align}
\hat{c}(d)&=  d e^{\mu h}
              - q_0 \rho e^{\rho h}
              - \frac{1}{2} \sigma^2 \frac{e^{\mu h}}{\rho-\mu}(m-1) d,
\end{align}
with
\begin{align}
m = \frac{1}{ \sigma ^2}
     \left( \textstyle \sqrt{\left(\mu - \frac{1}{2} \sigma ^2 \right)^2+ 2 \rho  \sigma ^2}
          - \left( \mu - \frac{1}{2} \sigma ^2  \right) 
     \right).
\end{align}
Further, \eqref{ass:rho} implies $m>2$.

\subsection{Comparative statics}

Note that
\begin{align}
\hat{c}(d)= A d -  q_0 \rho e^{\rho h},
\text{ with }
A = \frac{1}{2}
\frac{e^{\mu h}}{\rho-\mu}
\left( 2 \rho - \mu + \frac{1}{2} \sigma^2
       - \sqrt{\left(\mu-\frac{1}{2}\sigma^2 \right)^2+ 2\rho \sigma^2}
\right).
\end{align}

The next result analyzes the sensitivity of the boundary, and thus of the action region
with respect to the parameters of the model.
\begin{proposition}
The boundary in the GBM case has the following properties:
\begin{enumerate}
\item $\frac{\partial \hat{c}(d)}{\partial q_0} <0$
\item $A >0$
\item $\frac{h}{A} \frac{\partial A}{\partial h} = \mu h $, and it has the sign of $\mu$
\item $\frac{\sigma}{A} \frac{\partial A}{\partial \sigma} =
        - \frac{\sigma ^2}
               {\sqrt{\left( \mu-\frac{1}{2}\sigma^2\right)^2 + 2 \rho \sigma^2}}<0$
\item $\frac{\mu}{A} \frac{\partial A}{\partial \mu} =
        \mu h +
        \frac{1}{2}
        \frac{\mu}{\rho-\mu}
           \left(1- \frac{\mu+\frac{1}{2}\sigma^2 }{\sqrt{\left(\mu-\frac{1}{2}\sigma^2 \right)^2+ 2\rho \sigma^2}}
           \right) 
      $, and it has the sign of $\mu$
\item $\frac{\rho}{A} \frac{\partial A}{\partial \rho} =
       \frac{1}{2}
       \frac{\rho \sigma^2+  \mu^2 - \frac{1}{2} \mu \sigma^2 - \mu \sqrt{\left(\mu-\frac{1}{2}\sigma^2 \right)^2+ 2\rho \sigma^2}}
       { (\rho-\mu) \sqrt{\left(\mu-\frac{1}{2}\sigma^2 \right)^2+ 2\rho \sigma^2}} >0
      $
\end{enumerate}
\end{proposition}

\begin{proof}
Properties 1, 3, and 4 are immediate.

The other properties involve the same square root for the denominator. 
The signs are determined in all of the cases by showing that the numerators
can be rearranged and simplified to show that their signs depend only
on the sign of $\rho (\rho-\mu)$, which is positive given \eqref{ass:rho}.
These determinations ensure that the terms have 
the same sign for all of the relevant parameters.
\end{proof}

Property 2 says that the investment is responsive to the current demand.
Property 3 shows the importance of $\mu$:
when, e.g., $\mu>0$, a longer delay means above all a higher future demand,
hence a higher investment.
Property 4 confirms that more uncertainty makes the investor more cautious.
A similar logic explains property 5.

If the focus is on the precautionary bias only, then 
\begin{align}
b_{\sigma}(d) = \frac{1}{2}
\frac{e^{\mu h}}{\rho-\mu}
\left(\textstyle \sqrt{\left(\mu-\frac{1}{2}\sigma^2 \right)^2+ 2\rho \sigma^2}
       - \left( \mu + \frac{1}{2} \sigma^2 \right)\right) \, d >0.
\end{align}
But,
\begin{align}
\frac{\mu}{b_{\sigma}(d)} \frac{\partial b_{\sigma}(d)}{\partial \mu} =
        \mu
        \left( h 
        -\frac{1}{2} \frac{1}{\rho-\mu} \textstyle
          \frac{2\rho -\mu  +\frac{1}{2}\sigma^2 -\sqrt{\left(\mu-\frac{1}{2}\sigma^2 \right)^2+ 2\rho \sigma^2}}{\sqrt{\left(\mu-\frac{1}{2}\sigma^2 \right)^2+ 2\rho \sigma^2}} \right).
\end{align}
In the second factor, the first term is positive and the second one is negative.
We take $\mu>0$ for the discussion.
The overall sign of the elasticity depends, for example, on $h$:
if $h$ is small, then the elasticity is negative 
(the precautionary bias decreases as $\mu$ increases);
if $h$ is big, then the elasticity is positive 
(the precautionary bias increases).

Property 6 shows that the discount rate has two clear antagonistic effects:
the discounting bias increases in absolute value with respect to $\rho$ because
the benefits of investment are discounted, and the precautionary bias
decreases in absolute value because the future costs are discounted.
Thus, 
\begin{align}
\frac{\rho}{b_{\sigma}(d)} \frac{\partial b_{\sigma}(d)}{\partial \rho} =
       \frac{1}{2} \frac{\rho}{\rho-\mu} \textstyle
       \frac{\mu + \frac{1}{2}\sigma^2 - \sqrt{\left(\mu-\frac{1}{2}\sigma^2 \right)^2+ 2\rho \sigma^2}}
       {\sqrt{\left(\mu-\frac{1}{2}\sigma^2 \right)^2+ 2\rho \sigma^2}} <0.
\end{align}

\subsection{Simulations}
If $\rho=0.08$, $\mu = 0.03$, and $\sigma = 0.1$, then
the elasticity of $A$ w.r.t.\ to $h$ is $3\%$
at $h=1$ and $24\%$ at $h=8$.
The elasticity of $A$ w.r.t.\ $\sigma$ is  $- 21\%$ whatever the $h$.
The elasticity of $b_{\sigma}(d)$ w.r.t.\ to $h$ is $3\%$
for $h=1$ and $24\%$ for $h=8$.
The elasticity of $b_{\sigma}(d)$ w.r.t.\  $\sigma$ is  $153\%$
whatever the $h$.

Figure \ref{fig:geometric-behavior-1} 
shows a trajectory of demand for $h=8$
and $\sigma = 0.06$ with a starting point of $d=10^3$.
The committed capacity stops growing during the episode where demand is (fortuitously)
stabilized. Given the long delay, the committed capacity is almost always ahead
of the demand.

\begin{figure}[ht!]
\begin{center}
\includegraphics[width=0.35\paperwidth]{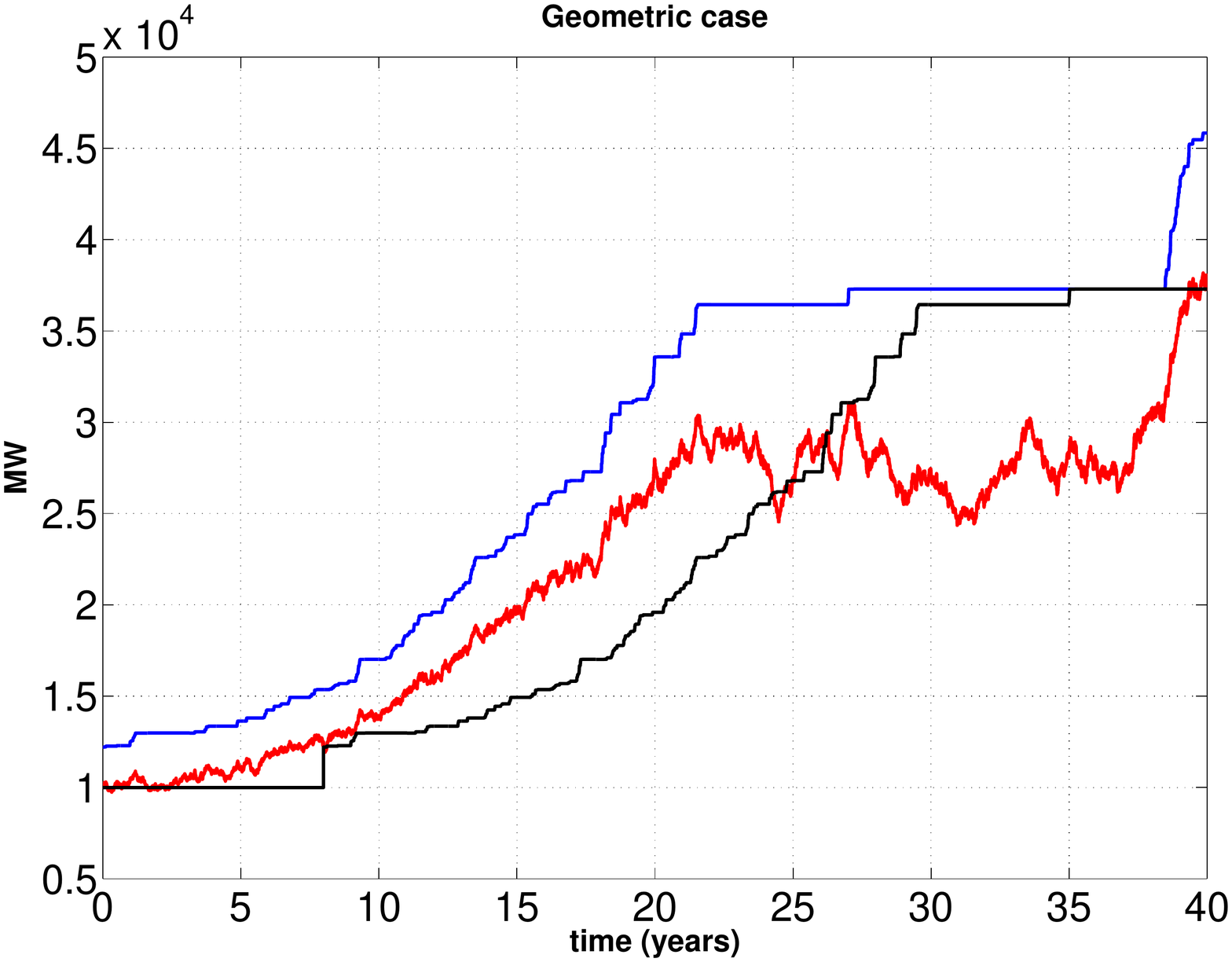}
\caption{
         Demand, committed, and installed capacity behavior in the geometric case
                 for $h=8$ years and when $\sigma = 0.06$.
         }
    \label{fig:geometric-behavior-1}
\end{center}
\end{figure}

\section{CIR model} \label{sec:CIR}

\subsection{The boundary}

For the case where the demand follows a Cox-Ingersoll-Ross model:
\begin{align}
\ud D_t=\gamma(\delta-D_t){{\ud}t}+\sigma \sqrt{D_t}\ud W_t, \qquad \gamma,\delta,\sigma>0,
\end{align}
then, under the assumption  $2\gamma\delta\geq \sigma^2$,
we have $\mathcal{O}=(0,+\infty)$.
We suppose that this assumption is true.
Also in this case \eqref{growth2} is verified with $\kappa_1=\varepsilon$ for any $\varepsilon>0$.
Therefore, according to \eqref{rhoK1}, we assume that $\rho>0$.

This case has
\begin{align}
 \beta_0(d)= e^{-\gamma h}d +   (1-e^{-\gamma h}) \delta
 \quad \text{ and } \quad
 \beta(d)  = e^{-\gamma h}\,\frac{d-\delta}{\rho+\gamma}+\frac{\delta}{\rho}.
\end{align}
Moreover,
\begin{align}
 [\mathcal{L}\phi](d)=\rho \phi(d)-\gamma(\delta-d)\phi'(d)-\frac{1}{2}\sigma^2 d\phi''(d),
 \quad \phi\in C^2(\mathcal{O};\mathbb{R}),
\end{align}
and the increasing fundamental solution to $\mathcal{L}\phi=0$ is
\begin{align}
\psi(d)=M(\rho/\gamma,2\gamma\delta/\sigma^2,2\gamma d/\sigma^2),
\end{align}
where $M$ is the confluent hypergeometric function of the first type.%
\footnote{See \citet{Abramowitz65}.}

Hence,
\begin{align}
\hat{c}(d) =& e^{-\gamma h} d +  (1-e^{-\gamma h}) \delta
               - q_0 \rho e^{\rho h}
               -\frac{1}{2} \sigma^2\frac{e^{-\gamma h}}{\rho +\gamma}
                \frac{\psi''(d)}{\psi'(d)} \\
=& \delta + e^{- \gamma h} (d-\delta)
-   q_0 \rho e^{\rho h}
- e^{- \gamma h } \,
\frac{\sigma ^2}{2 \gamma  \delta +\sigma ^2}  \,
\frac{M\! \left(2+\frac{\rho }{\gamma},2+\frac{2 \gamma \delta }{\sigma^2},\frac{2 d \gamma}{\sigma^2}\right)}
      {M\! \left(1+\frac{\rho}{\gamma},1+\frac{2 \gamma \delta}{\sigma^2},\frac{2 d \gamma}{\sigma^2}\right)}
 \,  d.
\end{align}

\subsection{Comparative statics}

The analysis is done with a stylized version of the boundary
based on the following results.
\begin{proposition} The boundary in the CIR model verifies:
\begin{enumerate}
\item Tangent at $d=0$:
\begin{align}
\mathrm{Tangent}(d)=
\frac{ \gamma  \delta }{\gamma  \delta +\frac{\sigma ^2}{2}}  e^{-\gamma h} d
+\left( 1-e^{-h \gamma } \right) \delta
-q_0 \rho e^{h \rho }
\end{align}
\item Asymptote when $d \rightarrow \infty$:
\begin{align}
\mathrm{Asymptote}(d)=
\frac{\rho}{\rho+\gamma} e^{-\gamma h}d+
\left(1-\frac{\rho}{\rho+\gamma}e^{-\gamma h}\right) \delta
- \frac{\sigma^2}{2 \gamma} \frac{\rho}{\rho+\gamma}e^{-\gamma h}
- q_0 \rho e^{\rho h}
\end{align}
\item The intersection between the two lines above is
\begin{align}
\left( \delta + \frac{\sigma^2}{2 \gamma}, \delta - q_0 \rho e^{\rho h}  \right)
\end{align}
\end{enumerate}
\end{proposition}

\begin{proof}
The calculation of the tangent is immediately given by the series expansion of $M$:
\begin{align}
\mathop{M\/}\nolimits\!\left(a,b,z\right)=\sum_{s=0}^{\infty}\frac{\left(a%
\right)_{s}}{\left(b\right)_{s}s!}z^{s}=1+\frac{a}{b}z+\frac{a(a+1)}{b(b+1)2!} z^2+ \cdots
\end{align}

To calculate the asymptote, we start from  \eqref{hatcmr}.
Let  $M(a,b;z)$ be the confluent hypergeometric function of the first type with parameters $a,b$.
Then
\begin{enumerate}
\item[$(i)$]  $zM'(a,b;z)=a(M(a+1,b;z)-M(a,b;z))$ (here $M'$ is the derivative w.r.t.\ $z$)
\item[$(ii)$] $M(a,b;z)\sim \frac{\Gamma(b)}{\Gamma(a)} e^z z^{a-b},$ when $z\rightarrow \infty$
\end{enumerate}
Using $(i)$,
\begin{align}
\frac{M(a,b;z)}{zM'(a,b;z)}=\frac{M(a,b;z)}{zM'(a,b;z)}=
\frac{M(a,b;z)}{a(M(a+1,b;z)-M(a,b;z))}=
\frac{1}{a\left(\frac{M(a+1,b;z)}{M(a,b;z)}-1\right)},
\end{align}
and using $(ii)$, we get
\begin{align}
\lim_{z\rightarrow\infty} \frac{M(a,b,z)}{zM'(a,b;z)}=0.
\end{align}
Thus, the slope of the asymptote of $\hat{c}$ is
\begin{align}
\alpha:= \lim_{d\rightarrow\infty }\frac{\hat{c}(d)}{d}= \lim_{d\rightarrow\infty }\frac{\rho\beta(d)}{d}=\frac{\rho}{\rho+\gamma} e^{-\gamma h}.
\end{align}
The calculation is:
\begin{align}
\kappa:=\lim_{d\rightarrow \infty} \hat{c}(d)-\alpha d.
\end{align}
Therefore,
\begin{align}
\kappa =
\delta\left(1-\frac{\rho}{\rho+\gamma}e^{-\gamma h}\right)
-\kappa_1 \frac{\rho}{\rho+\gamma}e^{-\gamma h}
- q_0 \rho e^{\rho h},
\end{align}
where
\begin{align}
\kappa_1:= \lim_{d\rightarrow \infty}\frac {\psi(d)}{\psi'(d)}.\end{align}
To compute the latter, $(i)$ is used to get
\begin{align}
\frac{M(a,b;z)}{M'(a,b;z)}=
\frac{z}{a\left(\frac{M(a+1,b;z)}{M(a,b;z)}-1\right)}.
\end{align}
Then, the use of $(ii)$ and $a \Gamma(a)=\Gamma(a+1)$ gets
\begin{align}
\lim_{z\rightarrow\infty}
\frac{M(a,b;z)}{M'(a,b;z)}=\lim_{z\rightarrow\infty}\frac{z}{z-a} =1.
\end{align}
Thus, given the function of interest $M(\rho/\gamma,2\gamma\delta/\sigma^2,2\gamma d/\sigma^2)$,
$\kappa_1= \frac{\sigma^2}{2 \gamma}$ is obtained.
The expression of the asymptote follows.

The expression of the intersection is a direct implication of points 1.\ and 2.\ of this
Proposition.
\end{proof}

For the economic interpretations, $\hat{c}(d)$ has the stylized expression:
\begin{align}
\min  \left\{ \text{Tangent}(d), \text{Asymptote}(d) \right\}.
\end{align}
The kink point
$\left( \delta + \frac{\sigma^2}{2 \gamma}, \delta - q_0 \rho e^{\rho h}  \right)$
is close to $( \delta , \delta )$ if the uncertainty is
small compared to the convergence speed.

When $h$ and $\sigma$ are small, the tangent is the 45 degree line
minus the discounting bias: committed capacity follows demand.
The asymptote is conservative because the capacity increases by 
only $\frac{\rho}{\rho + \gamma}$ for each unit increase of demand.

With a large convergence speed compared to the volatility (a small $\sigma^2/\gamma$),
the uncertainty has a negligible impact on the boundary.

The tangent and the asymptote become flatter and flatter as $h$ increases:
current conditions as measured by $d$ matter less when the delay is longer.
The flattening effect is exponential.
Reversion to the mean implies that as the delay increases,
the current demand progressively loses relevance for the prediction of the future demand.
No precautionary bias is needed at the limit for the large delays.

\subsection{Simulations}

The CIR model provides a rich setting to analyze the effects of the time-to-build, of the volatility, and of different convergence rates.

The following reference parameters are:
the initial value demand is 10,
the discount rate is $\rho = 0.08$,
the long-term demand is $\delta = 20$, and
the demand approaches this limit at a speed $\gamma=0.8$,

The alternative scenarios consider the four cases
where the delay $h=1$ or $8$, and the demand volatility $\sigma =0.1$ or $0.05$.
Figure \ref{fig:CIR-behavior-1} (Left) gives the four boundaries.
However, the two boundaries with $h=8$ are almost completely flat and confounded.
The other two have very close tangents and asymptotes and are hard to
discern visually.
The 45 degree line is also drawn.

Figure \ref{fig:CIR-behavior-1} (Right) shows a trajectory for $h=8$ and $\sigma = 0.2$.
The committed capacity is immediately at the maximum and then varies very little
except when the demand becomes exceptionally high for the first time.

\begin{figure}[ht!]
\begin{center}
\includegraphics[width=0.35\paperwidth]{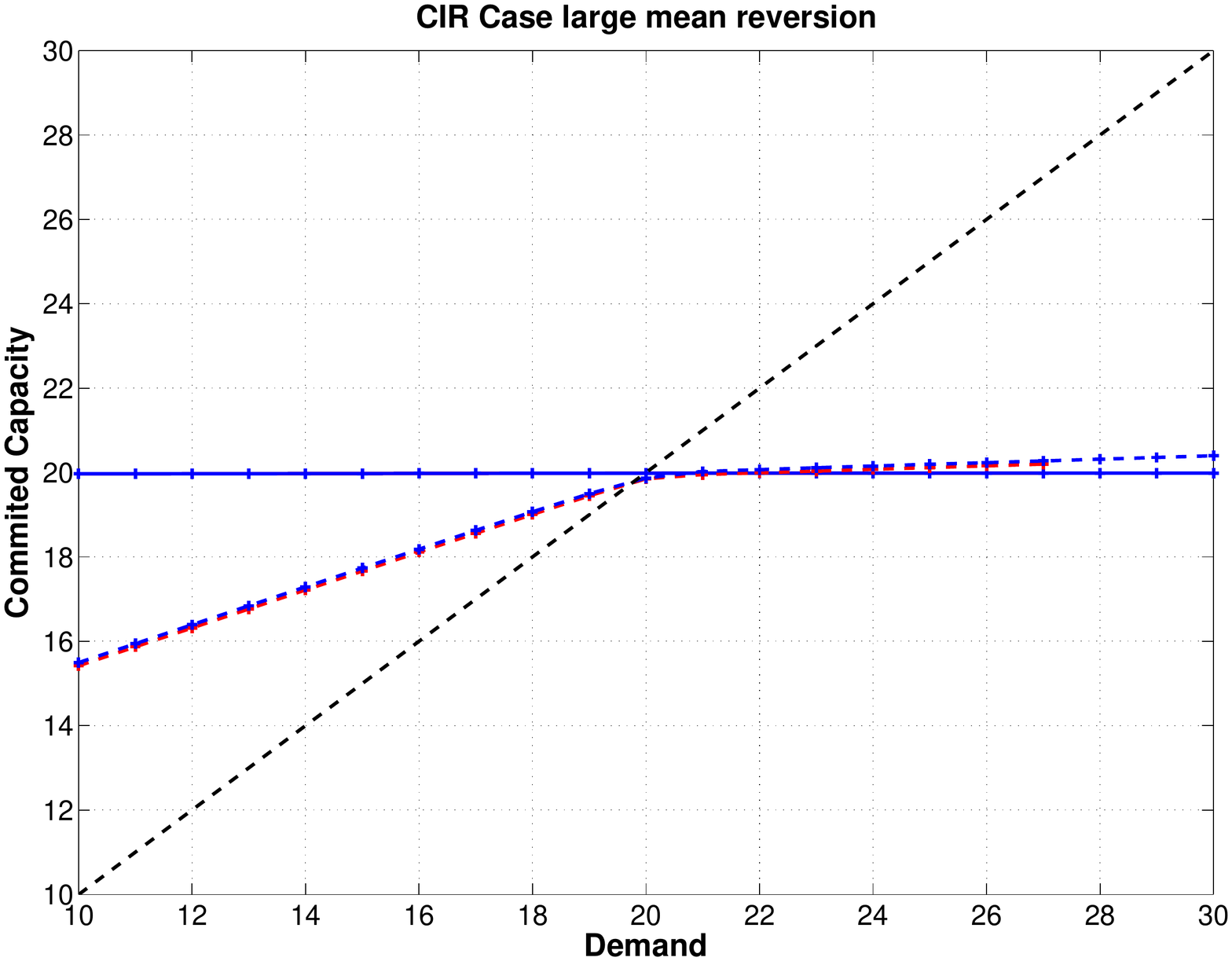}
\includegraphics[width=0.35\paperwidth]{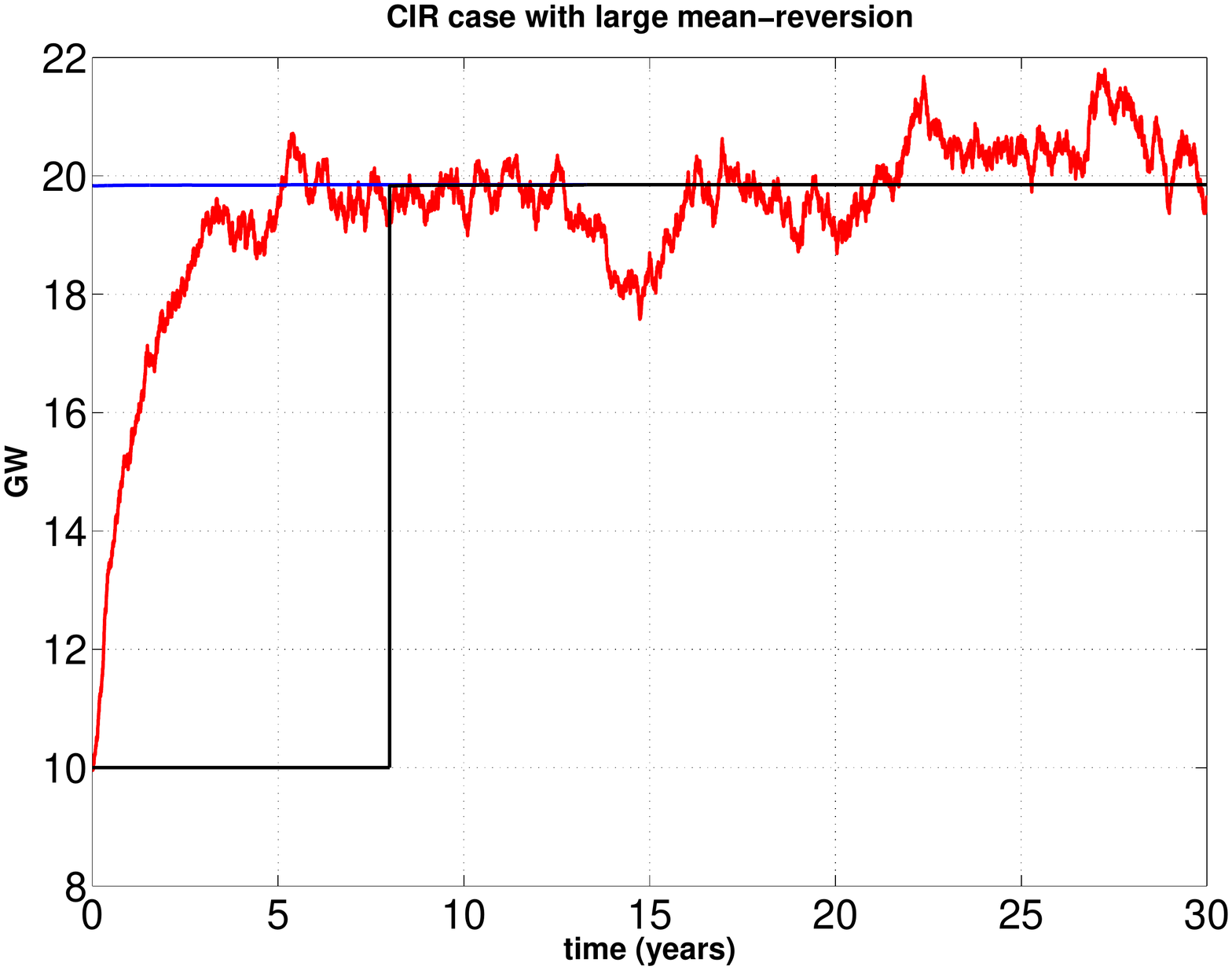}
\caption{(Left)  Investment boundaries.
         (Right) Demand, committed and installed capacity behavior with the CIR model
                 for an eight-year delay and a large mean-reversion ($\gamma=0.8$).
         }
    \label{fig:CIR-behavior-1}
\end{center}
\end{figure}

Figure \ref{fig:CIR-behavior-2} (Left) shows four boundaries with the same
parameters as in Figure \ref{fig:CIR-behavior-1} except that $\gamma=0.08$.
The boundaries have a less marked kink than with a faster convergence rate: boundaries
are more like the 45 degree line because the demand evolves much more slowly,
and they are much more alike in terms of positions and slopes.

Figure \ref{fig:CIR-behavior-2} (Right) shows a trajectory for $h=8$ and $\sigma = 0.2$.
The committed capacity is more responsive to the current conditions because
they are better predictors of the future demand than when $\gamma$ is large.
This effect plays for demand levels below 20 or above.

\begin{figure}[ht!]
\begin{center}
\includegraphics[width=0.35\paperwidth]{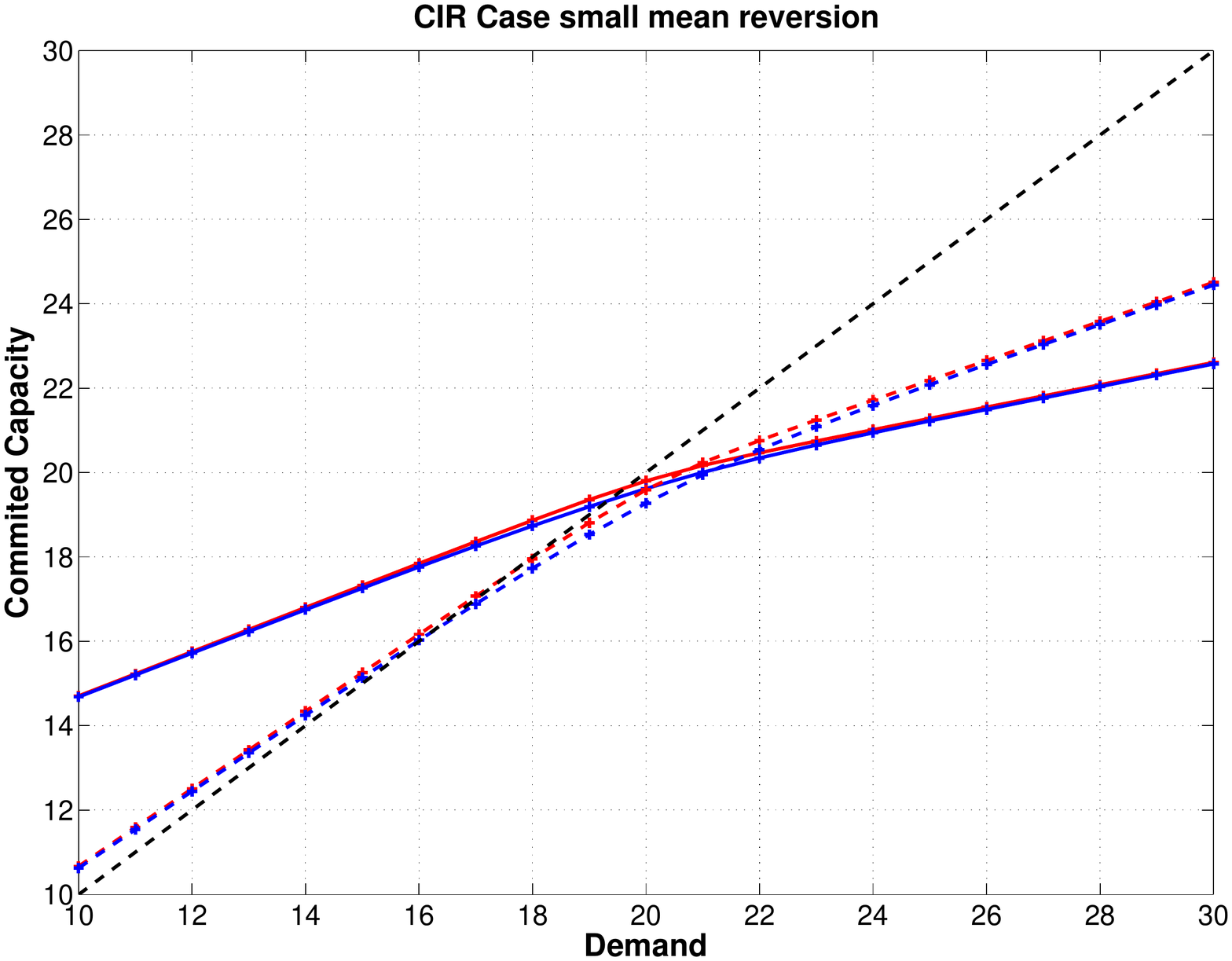}
\includegraphics[width=0.35\paperwidth]{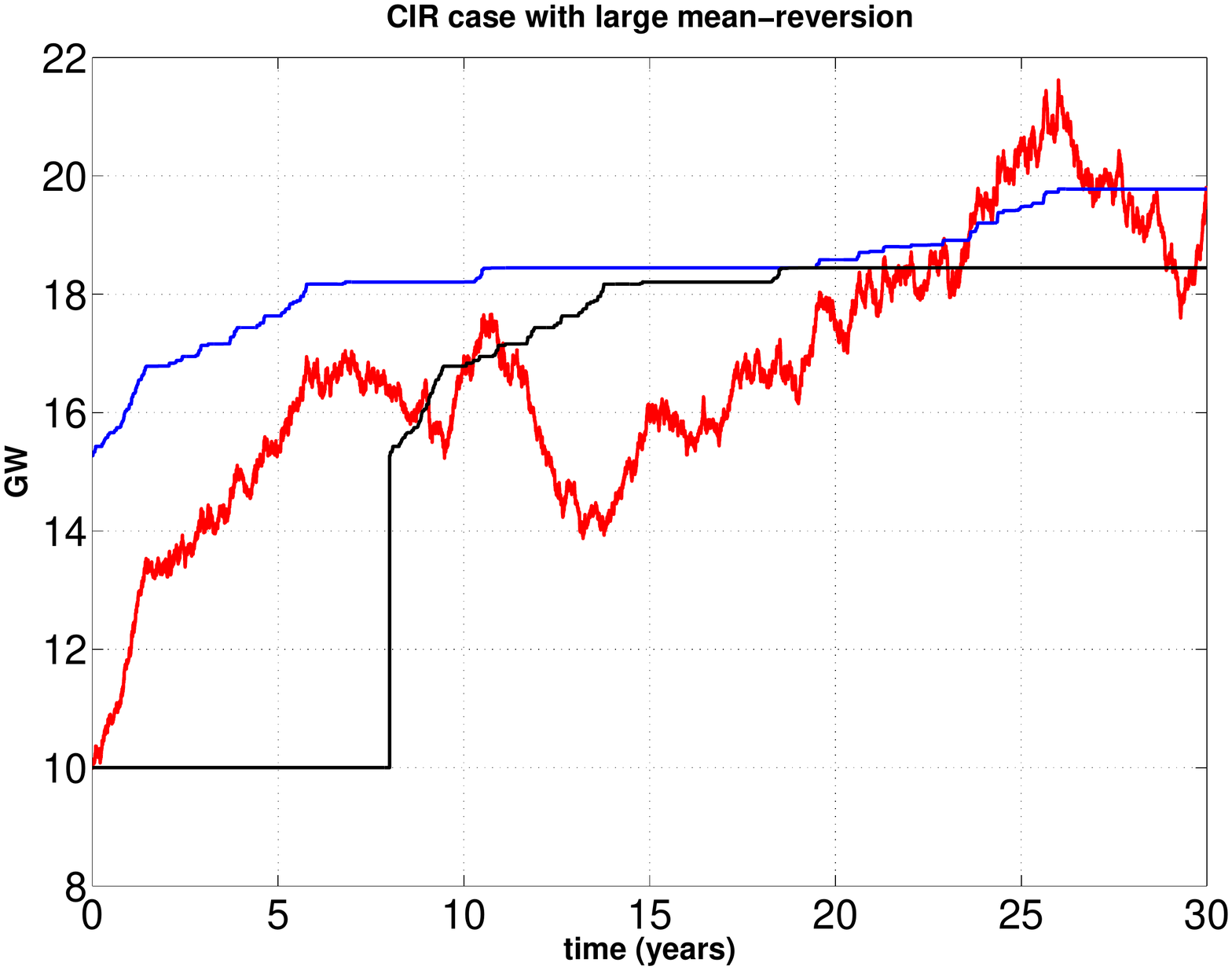}
\caption{(Left)  Investment boundaries.
         (Right) Demand, committed and installed capacity behavior in the CIR case
                 for an eight-year delay and a small mean-reversion ($\gamma=0.08$).  }
    \label{fig:CIR-behavior-2}
\end{center}
\end{figure}

\section{Conclusion} \label{sec:conclusion}

Electricity demand has a random part and is price sensitive.
Our minimization of an expected quadratic loss is founded on microeconomic theory,
and our optimal solution can be implemented as a competitive equilibrium.
In this paper where the delay between the investment decision and activation of the new capacity
is accounted for, we have characterized the explicit decision rules for important classes of demand processes.

The benefits of closed-form solutions cannot be overstated,
because we can show the interaction, in investors decisions, 
between the time-to-build and the uncertainty.
In particular, we identify the base rule and the two corrective terms:
the investor should invest if his or her committed capacity
(i.e., the capacity in the pipeline)
is below the best linear estimate of the future demand,
the given demand today, and the delay
minus a discounting bias 
and a precautionary bias determined by uncertainty and global risk aversion.
The latter term varies substantially with the demand model.

In the arithmetic Brownian motion, 
the delay and the uncertainty have additive separate effects.
In the geometric Brownian motion, the shocks are amplified exponentially
so that with a longer delay, restricting the future capacity becomes more costly. 
On the other hand, the discounting bias is accentuated by the delay. 
The question of which of these opposite effects dominates the other as the delay increases can be addressed with our explicit expressions.
In the CIR case, reversion to the mean implies that as the delay increases,
the current demand progressively loses relevance for the prediction of the future demand.
No precautionary bias is needed at the limit for the large delays.

\appendix

\section{Arithmetic Brownian Motion}  \label{sec:ABM}

\subsection{The Frontier}
With an arithmetic Brownian model of demand,
our model  is a particular case of \citet{Bar-Ilan02},
where the fixed investment cost is null.
The optimal strategy is simpler. The demand dynamics are:
\begin{align}
\ud D_t=\mu{{\ud}t}+\sigma \ud W_t, \qquad \mu\in\mathbb{R},\, \sigma>0,
\end{align}
then $\mathcal{O}=\mathbb{R}$ and \eqref{growth2}
is verified with $\kappa_1=\varepsilon$ for each $\varepsilon>0$.
Therefore, according to \eqref{rhoK1}, we assume that $\rho>0$.
Thus,
\begin{align}
 [\mathcal{L}\phi](d)=\rho \phi(d)-\mu \phi'(d)-\frac{1}{2}\sigma^2 \phi''(d),
 \quad \phi\in C^2(\mathcal{O}).
\end{align}
The increasing fundamental solution to $\mathcal{L}\phi=0$ is
$\psi(d)=e^{\lambda d}$
where $\lambda$ is the positive solution to
\begin{align}
 \rho-\mu \lambda -\frac{1}{2} \sigma^2 \lambda^2=0.
\end{align}
Because, in this case,
\begin{align}
\beta_0(d)= d + \mu h
\quad \text{ and } \quad
\beta(d) =  \frac{\mu h}{\rho}+\frac{d}{\rho}+\frac{\mu}{\rho^2}.
\end{align}
Due to Theorem \ref{prop:structure}, $\hat{c}$ is affine:
\begin{align}        \label{hatc2}
  \hat{c}(d)&=
  d+\mu h
  - q_0 \rho e^{\rho h}
  -\frac{\sqrt{\mu ^2+2 \rho  \sigma ^2} -\mu }{2 \rho }.
\end{align}

\subsection{Comparative statics}

Consider that
\begin{align}
\frac{\partial^2 \hat{c}(d)}{\partial h \partial \sigma} =0.
\end{align}
Whatever the time to build $h$, the investment is retarded in the same way by an increase in $\sigma$,
and conversely.
This additive separability makes it difficult to find the cross effects between the uncertainty and the delay with this model, 
contrary to \citet{Bar-Ilan02}.

An increase in uncertainty always retards investment:
\begin{align}
\partial \hat{c}(d) /\partial \sigma = - \frac{\sigma}{\sqrt{\mu^2 + 2 \rho \sigma^2}}
<0.
\end{align}

The variation of $\hat{c}(d)$ with respect to the time-to-build  $h$ is
\begin{align}
\partial \hat{c}(d) /\partial h = \mu - q_0 \rho^2 e^{\rho h }. 
\end{align}
The effect is to hasten investment if  $\mu$ is relatively large.
If $h$ is relatively large, then the cost of investment appears large compared
to the future discounted damage, and investment is retarded.
We retrieve the effects encountered in the case of the geometric Brownian motion.

Furthermore,
\begin{align}
\partial \hat{c}(d) /\partial \mu = 
h+
\frac{1}{2 \rho}
\left( \textstyle 1-\frac{\mu }{\sqrt{\mu ^2+2 \rho  \sigma ^2}} \right)
>0,
\end{align}
and
\begin{align}
\partial \hat{c}(d) /\partial \rho =
- q_0  (1+h \rho ) e^{\rho h}
+\frac{1}{2}
\left(
\frac{\mu ^2+\rho  \sigma ^2-\mu  \sqrt{\mu ^2+2 \rho  \sigma ^2}}{\rho ^2 \sqrt{\mu ^2+2 \rho  \sigma ^2}}
\right).
\end{align}
In the latter expression,
the first term is negative (the discounting bias is reinforced),
whereas
the second term is positive (the precautionary bias is attenuated).
Thus, we get the same effects encountered in the case of the geometric Brownian motion.

\subsection{Simulations}

On Figure~\ref{fig:arithmetic-b-cost} (Left),
$b:= \hat{c}(d)- d$ is given as a function of $\sigma$,
for two contrasted values of $h$ ($1$ and $8$ years).
The other parameters are:
$\rho= 0.08$, $\mu =300$, with an initial demand of 10,000~MW and
demand, committed capacity, and installed capacity all equal at date $0$
($D_0=C_0=K_0$).

Figure~\ref{fig:arithmetic-b-cost} (Left) shows that
the impact of the time-to-build  with these values 
is much more important than the impact of uncertainty.
Indeed, numerically,
$\partial \hat{c}(d) /\partial h$ is of the order of 300
whereas
$\partial \hat{c}(d) /\partial \sigma$ is of the order of $-1$.
The first effect largely dominates.

By and large, this result is in line with \citet{Bar-Ilan02}.
In their setting, increasing the time-to-build from one year to eight years
reverses the relation between uncertainty and investment,
which is possible only because they are not separable.
Specifically, for a long delay, 
an increase in uncertainty hastens investments but decreases their level.
But, these effects are very small (\citet[pp. 85, Figure~2]{Bar-Ilan02}).

The excess of committed capacity does not imply
that the system will hold an excess of installed capacity.
In fact, the reverse is observed in Figure~\ref{fig:arithmetic-b-cost} (Right).
In the case of a delay of eight years, an excess of committed capacity as measured
by the value of $b$ is 1,873~MW.
But in eight years, the demand will grow on average 2,400~MW,
which clearly indicates that the optimal strategy is to avoid excess installed capacity.

\begin{figure}[ht!]
\begin{center}
\includegraphics[width=0.35\paperwidth]{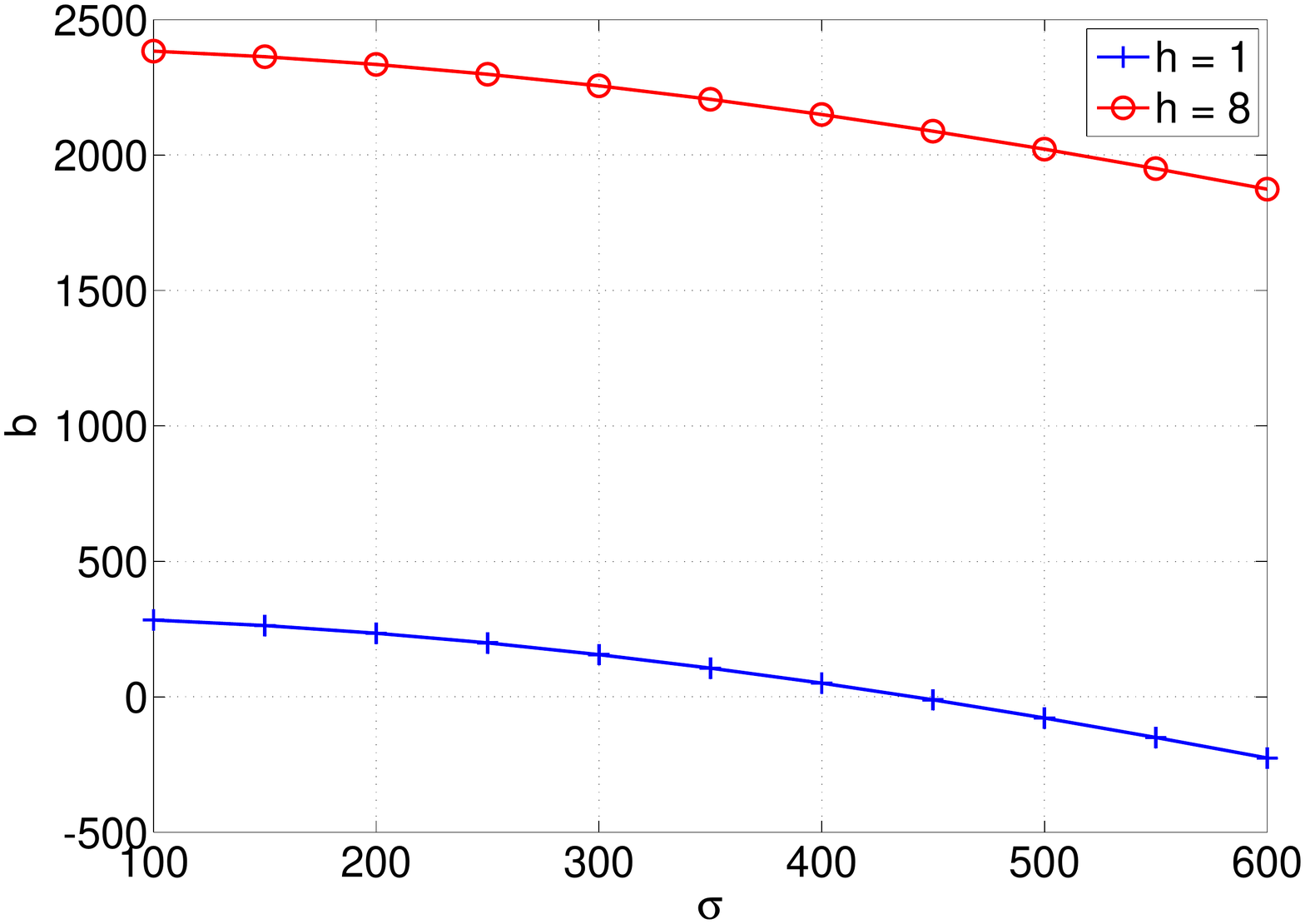}
\includegraphics[width=0.35\paperwidth]{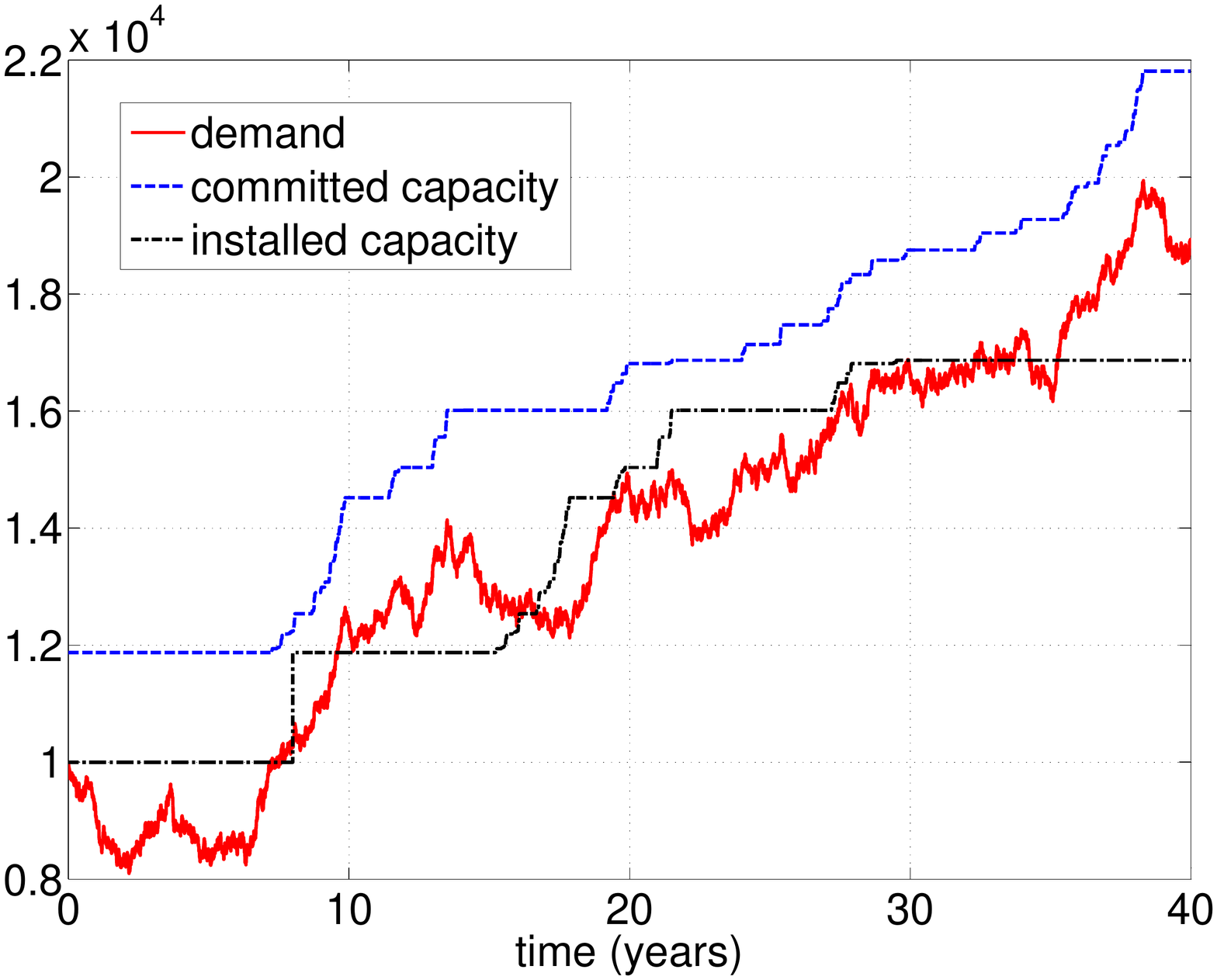}
\caption{
(Left)  $\hat{c}(d)-d$ as a function of $\sigma$ for two values of the time-to-build,
        $h=1$ (blue crosses) and $h=8$ (red circles).
(Right) Demand, committed and installed capacity behavior for an eight-year delay,
        $\sigma = 600$ MW$\cdot$year$^{1/2}$.}
      \label{fig:arithmetic-b-cost}
\end{center}
\end{figure}

\section{Proof of Proposition \ref{prop:competitive}} \label{sec:competitive}

Let $(c,d)\in\mathcal{S}$.
We prove first that
\begin{align}    \label{lemma1}
v_c(c,d)=&
\mathbb{E}\left[\int_0^{+\infty} e^{-\rho t} g_c(C^{c,*}_t,D^d_t)\ud t\right],
\end{align}
where $C^{*}$ is the optimal state process associated to the optimal control $I^*$
provided by Theorem \ref{prop:structure}.
Let $I^*\in\mathcal{I}$ be optimal for $(c,d)$.
Therefore,
\begin{align}          \label{mdf2}
\frac{G(c+\varepsilon,d;I^*)-G(c,d;\ I^*)}{\varepsilon}
\geq
\frac{v(c+\varepsilon,d)-v(c,d)}{\varepsilon}.
\end{align}
On the other hand,
\begin{align}          \label{mdf}
\frac{G(c+\varepsilon,d;I^*)-G(c,d:I^*)}{\varepsilon} =
\mathbb{E} \left[
           \int_0^{+\infty} e^{-\rho t}\
           \frac{g(C_t^{c,I^*}+\varepsilon,D^d_t)-g(C_t^{c,I^*},D^d_t)}{\varepsilon} \ud t
           \right].
\end{align}
Taking the limsup in \eqref{mdf2} and taking into account \eqref{mdf}, we get
\begin{align}           \label{plsd}
\limsup_{\varepsilon\downarrow 0} \frac{v(c+\varepsilon,d)-v(c,d)}{\varepsilon}
\leq
\mathbb{E} \left[ \int_0^{\infty} e^{-\rho t} g_c(C^{c,*}_t,D^d_t)\ud t \right].
\end{align}
On the other hand, arguing symmetrically with $c-\varepsilon$, we get
\begin{align}            \label{plsd2}
\liminf_{\varepsilon\downarrow 0} \frac{v(c,d)-v(c-\varepsilon,d)}{\varepsilon}
\geq
\mathbb{E} \left[ \int_0^{+\infty} e^{-\rho t} g_c(C^{c,*}_t,D^d_t)\ud t \right].
\end{align}
Therefore, \eqref{plsd} and \eqref{plsd2} assert \eqref{lemma1}.

Because of equation \eqref{lemma1}, $v_c\geq -q_0$, the definition of $\mathcal{A}$, and
because of
\begin{align}
g_c(C^{c,*}_t,D^d_t)= e^{-\rho h}\mathbb{E}\left[c-D^d_h\right]; 
\end{align}
the quadratic surplus application based on equation \eqref{eq:demand-price}
as an expected discounted revenue (price minus marginal cost) yields the result.

\begin{small}

\bibliographystyle{plainnat}

\end{small}

\end{document}